\documentclass[a4paper,12pt]{article}

\usepackage{amsthm}
\usepackage{amsmath}
\usepackage{epsfig}
\usepackage{amssymb}
\usepackage{latexsym}
\usepackage[T1]{fontenc}
\usepackage[utf8]{inputenc}
\usepackage{color}
\usepackage{url}

\newtheorem{definition}{Definition}
\newtheorem{theorem}[definition]{Theorem}

\newtheorem{lemma}[definition]{Lemma}
\newtheorem{conjecture}{Conjecture}

\let\epsilon=\varepsilon

\let\rho = \varrho

\begin{document}

\title{Cubic graphs with large circumference deficit}

\author{
Edita Máčajová${}^1$,
Ján Mazák${}^2$
\\[3mm]
\\{\tt \{macajova, mazak\}@dcs.fmph.uniba.sk}
\\[5mm]
${}^1$ Univerzita Komenského, Mlynská dolina, 842 48 Bratislava, Slovakia\\
${}^2$ Trnavská univerzita, Priemyselná 4, 918 43 Trnava, Slovakia
}

\maketitle

\begin{abstract}
The circumference $c(G)$ of a graph $G$ is the length of a longest cycle.
By exploiting our recent results on resistance of snarks, we construct infinite classes of cyclically $4$-, $5$- and $6$-edge-connected cubic graphs with circumference ratio $c(G)/|V(G)|$ bounded from above by $0.876$, $0.960$ and $0.990$, respectively.
In contrast, the dominating cycle conjecture implies that the circumference ratio of a cyclically $4$-edge-connected cubic graph is at least $0.75$.

In addition, we construct snarks with large girth and large circumference deficit,
solving Problem 1 proposed in [J. Hägglund and K. Markstr\"om,
On stable cycles and cycle double covers of graphs with large circumference, Disc. Math. 312 (2012), 2540--2544].
\end{abstract}


\medskip

{\bf Classification:} 05C15, 05C38

\section{Introduction}

A cycle is one of the most basic structures in a graph, so it comes as no surprise that cycles have been analysed from the very beginnings of graph theory.
This article focuses on longest cycles in cubic graphs. The {\it circumference\/} $c(G)$ of a graph $G$ is the length of a longest cycle. The {\it circumference ratio\/} is the ratio of circumference to order. The {\it circumference deficit\/} is the difference between order and circumference.

A lot of attention was given to Hamiltonian graphs, that is, graphs with zero circumference deficit. Compared to the vast tomes written on hamiltonicity, non-Hamiltonian graphs appear rather neglected despite there is plenty of investigation to be done.
The problem of determining the circumference of a given graph is NP-hard and even approximation is a very tough problem \cite{NP}, so no simple characterisations are expected. 

It transpired in many areas that the most interesting cubic graphs are those with chromatic index four. Such graphs are called {\it snarks}; we will additionally require snarks to have girth at least five and cyclic edge-connectivity at least four. We will encounter them in Section~\ref{snarks} where we prove the existence of a snark of girth at least $g$ and circumference deficit at least $g$ for every integer $g$.
Section~\ref{ratio} is devoted to upper bounds on circumference ratio; we provide linear bounds for certain classes of cyclically $4$-, $5$-, and $6$-edge-connected cubic graphs. These bounds serve as the best presently known upper bounds on general lower bounds on circumference ratio.

Each subcubic graph can be transformed into a $3$-edge-colourable graph by removing sufficiently many edges. The least number of edges that need to be removed is the {\it resistance} of the graph. There are snarks with arbitrarily large resistance (see e.g. \cite{steffen, oddness}). Most of our constructions are based on building blocks with large resistance; the usefulness of such blocks is demonstrated in Lemma~\ref{lemma1} which plays the key role in our proofs.

\section{Circumference ratio of cubic graphs}
\label{ratio}

Circumference ratio of cubic graphs strongly depends on connectivity. Since each vertex of a cubic graph is separated by three edges from the rest of the graph, the classical notions of vertex-connectivity and edge-connectivity are of limited use. A refined measure of connectivity, much more appropriate for our purpose, is provided by cyclic edge-connectivity. A graph is {\it cyclically $k$-edge-connected\/} if at least $k$ edges must be removed to disconnect it into components among which there are at least two containing a cycle. For cubic graphs, the notion of cyclic $k$-edge-connectivity coincides with $k$-vertex-connectivity and $k$-edge-connectivity for $k\in\{1,2,3\}$ \cite{atoms}.

If we allow bridges in our graphs, there are infinitely many trivial cubic graphs with circumference $5$.
What is more interesting, Bondy and Entringer \cite{be3} proved that every $2$-edge-connected cubic graph $G$ contains a cycle of length at least $4\log |V(G)|-4\log \log |V(G)|-20$. This bound is essentially best possible, as shown by Lang and Walther \cite{lw10}. 
Bondy and Simonovits \cite{bs5} conjectured the existence of a constant $c$ such that every $3$-connected cubic graph $G$ has circumference at least $|V(G)|^c$ and showed that $c \le \log_98\approx 0.946$. The conjecture was verified by Jackson \cite{jackson} for $c = \log (1+\sqrt5)-1\approx 0.694$; the constant $c$ has recently been improved to $0.753$ \cite{bbmy}. Bondy also conjectured the following.
\begin{conjecture}[{\cite[Conjecture 1]{jackson}}]
There exists a constant $c > 0$ such that every cyclically $4$-edge-connected cubic graph $G$ has circumference at least $c\,|V(G)|$.
\label{conj:bondy}
\end{conjecture}
The dominating cycle conjecture \cite{fleischner} implies that $c\ge 0.75$ since a dominating cycle in a cubic graph $G$ has length at least $0.75|V(G)|$.
In addition, Thomassen \cite{thomassen} conjectured that there exists an integer $k$ such that every cyclically $k$-edge-connected cubic graph is Hamiltonian. This would mean $c = 1$ for sufficiently connected cubic graphs.

We summarize the currently known results together with our contributions in Table~\ref{tabulka}. The column LB displays a lower bound which holds for all graphs; the column UB shows an upper bound on circumference for a certain infinite class of graphs with required cyclic edge-connectivity. All the bounds are asymptotic; each of them is expressed as a function of order $n$.

\begin{table}[!h]
\begin{center}
\renewcommand\arraystretch{1.2}
\begin{tabular}{|c|c|c|c|}
\hline
connectivity & LB & conjectured LB & UB\\\hline
2 & $\log n$ &  & $\log n$\\\hline
3 & $n^{0.753}$ & & $n^{0.946}$\\\hline
4 & $n^{0.753}$ & $0.75n$ & $0.875n$\ $\star$\\\hline
5 & $n^{0.753}$ & $0.75n$ & $0.960n$\ $\star$\\\hline
6 & $n^{0.753}$ & $0.75n$ & $0.990n$\ $\star$\\\hline
7+ & $n^{0.753}$ & $0.75n$\ \ \hbox{and}\ \ $n$ & $n$ \\\hline
\end{tabular}
\end{center}
\caption{Summary of results on circumference (our contribution are marked by $\star$).}
\label{tabulka}
\end{table}

The crucial observation used in our construction of graphs with large circumference deficit is captured in the following lemma.

\begin{lemma}
\label{lemma1}
Let $H$ be a subgraph of a bridgeless cubic graph $G$.
If $H$ has resistance $k$, then any cycle of $G$ not contained in $H$ misses at least $k$ vertices of $H$.
\end{lemma}

\begin{proof}
If a cycle $C$ of $G$ is not contained in $H$, then its intersection with $H$ is a union of vertex-disjoint paths.
Take each of those paths in turn and alternately colour the edges along the path by colours $1$ and $2$. Colour all the remaining edges by $3$.
There are two types of vertices in $H$: those that have their incident edges coloured by $1$, $2$, $3$ and those with all incident edges coloured by $3$. We call the vertices of the latter type {\it bad\/}. By removing all the bad vertices we obtain a $3$-edge-colourable graph from $H$.
However, $H$ has resistance $k$, and thus there are at least $k$ bad vertices.
Obviously, no bad vertex belongs to $C$, hence $C$ misses at least $k$ vertices of $H$.
\end{proof}

In order to construct infinite classes of graphs with circumference ratio promised in Table~\ref{tabulka}, we employ cubic construction blocks described in \cite{oddness}. (We do not include construction details or proofs of their properties in this article; an interested reader can find them in \cite{oddness}.)

\begin{figure}
\begin{center}
\includegraphics[scale=.8]{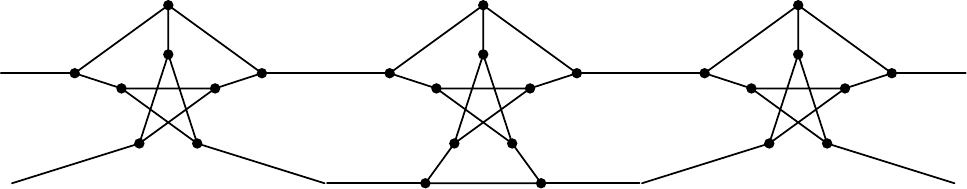}
\end{center}
\caption{The building block $N_2$ with $26$ vertices and resistance $2$.}
\label{fig:N2}
\end{figure}

The cubic graph $N_2$ has order $26$, resistance $2$, and two pairs of dangling edges (see Figure~\ref{fig:N2} and \cite[Section 7]{oddness}).
We create a graph $G$ by arranging $m\ge 2$ copies of $N_2$ along a circle and for each copy $K$, we connect one pair of dangling edges of $K$ to a pair of dangling edges of the following copy and the other pair of dangling edges to a pair in the previous copy. (The exact way of how we do it does not matter because we only need to preserve cyclic $4$-edge-connectivity.)
According to Lemma~\ref{lemma1}, a cycle of $G$ either belongs to one copy of $N_2$ or misses at least $2$ vertices in each of the $m$ copies of $N_2$, hence $G$ has circumference deficit at least $2m$ and circumference ratio at most $24m/26m = 12/13\approx 0.92$. A different idea used in Theorem~\ref{thm7/8} leads to a better upper bound for cyclically $4$-edge-connected graphs.

The construction described in the previous paragraph can be repeated with the cyclically $5$-edge-connected building block $Z$ with order $25$ and resistance $1$ (see \cite[Section 8 and Figure 5]{oddness} for a description of $Z$; this block was also used by Steffen under the name $T$ \cite[Theorem 2.3]{steffen}). The graph $Z$ has seven dangling edges naturally split into two triples and one single dangling edge. We repeat the circular construction with $m$ copies of $Z$ instead of $N_2$; the role of the pairs of dangling edges is now played by the triples. The single dangling edges are joined to a cycle of length $m$ (one dangling edge to each vertex of the cycle). The resulting graph has circumference ratio at most $24/25 = 0.96$ and is cyclically $5$-edge-connected.

A similar construction can also be used to construct cyclically $6$-edge-connected graphs; however, the details are more complicated. Section 9 of \cite{oddness} describes a cyclically $6$-edge-connected graph $M_r$ of order $99r$ with resistance at least $r$ for each even positive integer $r$, but there are no dangling edges in this graph, thus we cannot use it directly: we first have to cut a few suitable edges to obtain a block which would allow a construction of cyclically $6$-edge-connected graphs.

The graph $M_r$ is obtained by symmetrically applying superposition to a graph $L_r$ composed of $r$ circularly arranged isomorphic copies of the block $P_3$ (that is, the Petersen graph with one vertex removed). Therefore, $M_r$ also contains $r$ isomorphic blocks $A_1, A_2, \dots, A_r$ arranged along a circle. Each two consecutive blocks of $M_r$ are joined by three edges. We cut all the edges between $A_1$ and $A_2$ to form a cubic graph $M_r'$ with two triples of dangling edges. The graph $M_r'$ has order $99r$ and resistance at least $r-3$. (According to the definition of resistance, the removal of an edge can decrease resistance by at most $1$. The resistance of $M_r'$ is actually $r$, but that would require a detailed proof; we will use the obvious lower bound of $r-3$ here because it is sufficient for our purpose.)

Consequently, we can use $M_r'$ in place of $N_2$ in the above-described construction (with triples of dangling edges instead of pairs).
The resulting cyclically $6$-edge-connected cubic graph $G$ has order $m\cdot 99r$ and circumference deficit at least $m(r-3)$, thus its circumference ratio is at most $1-(r-3)/99r$. By taking a sufficiently large $r$ we can make this ratio to be arbitrarily close to $98/99\approx 0.990$.

\begin{theorem}
For each integer $m$, there exists a cyclically $4$-edge-connected cubic graph with order $8m$ and circumference $7m+2$.
\label{thm7/8}
\end{theorem}

\begin{proof}
Let $u$ and $v$ be two adjacent vertices of the Petersen graph $P$. We remove the path $uv$, but keep the dangling edges incident to exactly one of its endvertices; the dangling edges incident to $u$ are {\it input edges} and the dangling edges incident to $v$ are {\it output edges}. We say that a path {\it passes through\/} $B$ if it starts with a vertex incident to an input edge and ends in a vertex incident to an output edge. We say that a cycle passes through $B$ if a portion of this cycle (a path) passes through $B$. The resulting graph $B$ has two properties interesting to us.

First, if a path passes through $B$, it cannot pass through all the vertices of $B$:
otherwise we would be able to extend this path by $u$ and $v$ to a Hamiltonian cycle of $P$, but $P$ has no such cycle. 
Second, if we take two disjoint paths passing through $B$, there is at least one vertex of $B$ missed by both of these paths. Otherwise, we can extend the first path by $u$, extend the other path by $v$, and then concatenate them together by adding two edges to form a Hamiltonian cycle of $P$ which is a contradiction.

Let $G$ be the graph obtained from $m$ copies of $B$ arranged along a circle in such a way that the output edges of each copy are identified with the input edges of the following copy (see Figure~\ref{fig:chain}). The graph $G$ is cyclically $4$-edge-connected and has order $8m$.
\begin{figure}
\begin{center}
\includegraphics{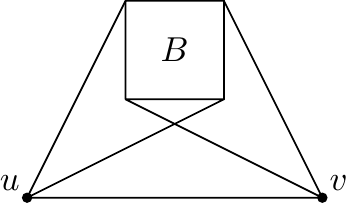}
\hskip 1.5cm
\raise 5mm\hbox{\includegraphics{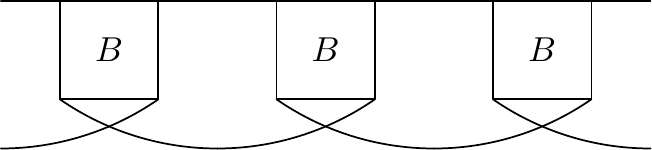}}
\end{center}
\caption{The building block $B$ and a sketch of the graph $G$.}
\label{fig:chain}
\end{figure}

Let $C$ be a cycle in $G$. Note that $C$ passes through each copy of $B$ at most twice. According to the two properties of $B$ proved above, no matter how many times $C$ passes through $B$, at least one vertex of $B$ is missed. The only possibility for $C$ to contain all vertices of $B$ is to enter by an input edge and then leave by the other input edge (of course, it can also both enter and leave by output edges, which is essentially the same situation). Consequently, the cycle $C$ misses at least one vertex in each of at least $m-2$ copies of $B$, and thus the circumference of $G$ is at most $8m-(m-2) = 7m+2$. Since $G$ contains a cycle of length $7m+2$, the derived upper bound on its circumference is tight.
\end{proof}

We propose the following strengthening of Conjecture~\ref{conj:bondy}.

\begin{conjecture}
Every cyclically $4$-edge-connected cubic graph has circumference ratio at least $7/8$.
\end{conjecture}

\section{Large girth and large circumference deficit}
\label{snarks}

This section is motivated by the cycle double cover conjecture (CDCC).
Huck \cite{huck} showed that the smallest possible counterexample to CDCC has girth at least $12$.
Brinkmann et al. \cite{bbb, hm} proved that if a bridgeless cubic graph $G$ has a cycle of length at least $|V(G)|-10$, then $G$ has a cycle double cover.
Put together, smallest counterexamples to CDCC can only be found in the class of snarks with girth at least $12$ and with circumference deficit at least $11$.
Since no such snark has been known before, the following problem is very relevant.

\medskip

\noindent {\bf Problem} (H\"agglund and Markstr\"om \cite{hm}). For each integer $g$,
construct a snark of girth at least $g$ and circumference deficit at least $g$.

\medskip

We solve this problem in Theorem~\ref{thm3}. The construction used in the proof of Theorem~\ref{thm3} can be modified to produce snarks with arbitrarily large girth and linear circumference deficit.

\begin{theorem}
For every integer $g$ there exists a snark with girth at least $g$ and circumference deficit at least $g$.
\label{thm3}
\end{theorem}

\begin{proof}
We construct the desired snark for every integer $g\ge 5$ which is enough to prove the theorem.

Let $H_0$ be a snark of girth at least $g$; the existence of such snarks has been proved by Kochol \cite{kochol}.
Let $v_1$ and $v_2$ be two adjacent vertices of $H_0$ and let $e_i$ and $f_i$, for $i\in\{1,2\}$,
be the two edges incident to $v_i$ and not incident to $v_{3-i}$.
Let $H_1$ be the cubic graph obtained from $H_0$ by removing $v_1$ and $v_2$ while keeping the dangling edges $e_1$, $e_2$, $f_1$, $f_2$.
The well-known parity lemma assures that the edges $e_1$ and $f_1$ have the same colour in every $3$-edge-colouring of $H_1$
(otherwise $H_0$ would be $3$-edge-colourable, but it is a snark).

\begin{figure}
\begin{center}
\includegraphics{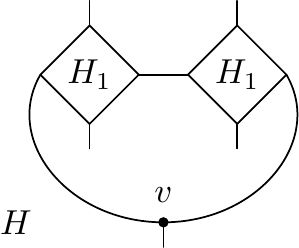}
\end{center}
\caption{The graph $H$.}
\label{fig:H_1}
\end{figure}

Take two copies of $H_1$ and join them as indicated in Fig.~\ref{fig:H_1} (the edge $f_1$ of the first copy is identified with the edge $f_1$ of the second copy and the edges $e_1$ of both copies are attached to an additional vertex $v$). If $H$ was $3$-edge-colourable, then the colour of $e_1$ of the first copy of $H_1$ would be the same as the colour of $f_1$ and, in turn, the same as the colour of $e_1$ of the second copy, leading to a contradiction at $v$. Hence, $H$ is not colourable and has resistance at least $1$. Moreover, there is no cycle of length less than $g$ in $H$, and any path with endvertices incident to dangling edges of $H$ passes through at least $g-1$ vertices of $H$.

Let $G$ be a cubic graph obtained from $g$ copies of $H$ arranged along a circuit in such a way that two dangling edges of a copy of $H$ are attached to the previous copy and two of them are attached to the next. The remaining $g$ edges can be joined to a cycle of length $g$ in an arbitrary way preserving maximum degree $3$.
The graph $G$ clearly has girth at least $g$, is cyclically $4$-edge-connected and is not $3$-edge-colourable. Any cycle of $G$ not contained in $H$ misses at least one vertex in each copy of $H$ thanks to Lemma~\ref{lemma1}, and thus $G$ has circumference deficit at least $g$.
\end{proof}

\noindent{\bf Acknowledgements.} This work was supported from the APVV grants APVV-0223-10 and ESF-EC-0009-10 within the EUROCORES Programme EUROGIGA (project GReGAS) of the European Science Foundation.


\begin{thebibliography}{99}

\bibitem{bbmy}
M. Bilinski, B. Jackson, J. Ma, X. Yu, Circumference of 3-connected claw-free graphs and large Eulerian subgraphs of 3-edge-connected graphs, J. of Comb. Theory Ser. B 101 (2011), 214--236.

\bibitem{be3}
J. A. Bondy and R. C. Entringer, Longest cycles in 2-connected graphs with prescribed maximum degree, Canad. J. Math. 32 (1980), 1325--1332.

\bibitem{bs5}
J. A. Bondy and M. Simonovits, Longest cycles in 3-connected cubic graphs, Canad. J. Math. 32 (1980), 987--992.

\bibitem{bbb}
G. Brinkmann, J. Goedgebeur, J. Hägglund, K. Markstr\"om, Generation and properties of snarks,
J. Comb. Theory Ser. B 103 (4) (2013), 468--488.

\bibitem{NP}
G. Chen, J. Xu, and X. Yu, Circumference of graphs with bounded degree, SIAM J. Comput. 33 (5) (2004), 1136--1170.

\bibitem{fleischner}
H. Fleischner, Some blood, sweat, but no tears in Eulerian graph theory, Congr. Numer. 63 (1988), 8--48.

\bibitem{hm}
J. Hägglund, K. Markstr\"om, On stable cycles and cycle double covers of graphs with large
circumference, Disc. Math. 312 (2012) 2540--2544.

\bibitem{huck}
A. Huck, Reducible configurations for the cycle double cover conjecture, in: Proceedings of the 5th Twente Workshop on Graphs and
Combinatorial Optimization (Enschede, 1997), vol. 99 (2000), 71--90.

\bibitem{jackson}
B. Jackson, Longest cycles in 3-connected cubic graphs, J. Comb. Theory Ser. B 41 (1986), 17--26.

\bibitem{kochol}
M. Kochol, Snarks without small cycles, J. Combin. Theory Ser. B 67 (1996), 34--47.

\bibitem{lw10}
R. Lang and H. Walther, \"Uber längste Kreise in regulären Graphen, in ``Beitrage zur Graphenteorie, Kolloquium, Manebach 1967'', Teubner, Leipzig (1968),  91--98.

\bibitem{oddness}
R. Lukoťka, E. Máčajová, J. Mazák, and M. Škoviera, Small snarks with large oddness, {\tt arXiv:1212.3641}.

\bibitem{atoms}
R. Nedela and M. Škoviera, Atoms of cyclic connectivity in cubic graphs, Math. Slovaca 45 (1995), 481--499.


\bibitem{steffen}
E. Steffen, Measurements of edge-uncolorability, Disc. Math. 280 (2004), 191--214.

\bibitem{thomassen}
C. Thomassen, personal communication, a conference in Vienna (1991).


\end{thebibliography}
\end{document}